\documentclass[a4paper, 12pt]{article}
\usepackage{amssymb, amsmath, amsthm}

\newcommand\R{\mathbb{R}}
\newcommand\C{\mathbb{C}}
\newcommand\N{\mathbb{N}}
\newcommand\B{\mathbb{B}}

\renewcommand\i{{\rm 1\kern -.3600em 1}}

\newcommand\vd{\delta}\newcommand\eps{\varepsilon}

\newtheorem{theorem}{Theorem}[section]
\newtheorem{corollary}[theorem]{Corollary}
\newtheorem{lemma}[theorem]{Lemma}
\newtheorem{proposition}[theorem]{Proposition}
\newtheorem{definition}[theorem]{Definition}
\newtheorem{example}[theorem]{Example}
\newtheorem{remark}[theorem]{Remark}

\numberwithin{equation}{section}
\allowdisplaybreaks[3]
\usepackage{color}

\begin{document}

\author{\textbf{Dmitri L.~Finkelshtein}\\
{\small Institute of Mathematics, Ukrainian National Academy of
Sciences, 01601 Kiev, Ukraine}\\
{\small fdl@imath.kiev.ua} \and
\textbf{Yuri G.~Kondratiev} \\
{\small Fakult\"at f\"ur Mathematik, Universit\"at Bielefeld,
D 33615 Bielefeld, Germany}\\
{\small Forschungszentrum BiBoS, Universit\"at Bielefeld, 
D 33615 Bielefeld, Germany}\\
{\small kondrat@mathematik.uni-bielefeld.de} \and
\textbf{Maria Jo\~{a}o Oliveira} \\
{\small Universidade Aberta, P 1269-001 Lisbon, Portugal}\\
{\small CMAF, University of Lisbon, P 1649-003 Lisbon, Portugal}\\
{\small oliveira@cii.fc.ul.pt}}

\title{Glauber dynamics in the continuum via generating functionals evolution}

\date{}
\maketitle

\begin{abstract}
We construct the time evolution for states of Glauber dynamics for a spatial 
infinite particle system in terms of generating functionals. This is carried 
out by an Ovsjannikov-type result in a scale of Banach spaces, leading to a 
local (in time) solution which, under certain initial conditions, might be 
extended to a global one. An application of this approach to Vlasov-type 
scaling in terms of generating functionals is considered as well.
\end{abstract}


\medskip


\newpage

\section{Introduction}

Originally, Bogoliubov generating functionals (shortly GF) were introduced by 
N.~N.~Bogoliubov in \cite{Bog46} to define correlation functions for 
statistical mechanics systems. Apart from this specific application, and many 
others, GF are, by themselves, a subject of interest 
in infinite dimensional analysis. This is partially due to the fact that to 
a probability measure $\mu$ defined on the space $\Gamma$ of locally finite 
configurations $\gamma\subset\R^d$ one may associate a GF
$$
B_\mu(\theta) :=\int_\Gamma d\mu(\gamma)\,\prod_{x\in \gamma }(1+\theta(x)),
$$
yielding an alternative method to study the stochastic dynamics of an 
infinite particle system in the continuum by exploiting the close relation 
between measures and GF \cite{FKO05,KoKuOl02}. 

Within the semigroups theory, a non-equilibrium Glauber dynamics has been  
constructed through evolution equations for correlation functions in 
\cite{FKK09,FKKZ2010,KoKtZh06}. However, within the GF context, semigroup 
techniques seem do not work. Alternatively, existence and uniqueness 
results for the Glauber dynamics through GF arise naturally from Picard-type 
approximations and a method suggested in \cite[Appendix 2, A2.1]{GS58} in a 
scale of Banach spaces (Theorem \ref{Th1}). This method, originally presented 
for equations with coefficients time independent, has been extended to an 
abstract and general framework by T.~Yamanaka in \cite{Y60} and 
L.~V.~Ovsjannikov in \cite{O65} in the linear case, and many applications were 
exposed by F.~Treves in \cite{T68}. As an aside, within an analytical 
framework outside of our setting, all these statements are very closely 
related to variants of the abstract Cauchy-Kovalevskaya theorem. However, all 
these abstract forms, namely, Theorem \ref{Th1}, only yield a local solution, 
that is, a solution which is defined on a finite time interval. Moreover, 
starting with an initial condition from a certain Banach space, in general the 
solution evolves on larger Banach spaces. It is only for a certain class of 
initial conditions that the solution does not leave the initial Banach space. 
In this case, the solution might be extended to a global solution (Corollary 
\ref{bdd}).    

As a particular application, this work concludes with the study of the 
Vlasov-type scaling proposed in \cite{FKK10a} for generic continuous particle 
systems and accomplished in \cite{FKK10} for the Glauber dynamics. The general 
scheme proposed in \cite{FKK10a} for correlation functions yields a limiting 
hierarchy which possesses a chaos preservation property, namely, starting with 
a Poissonian (non-homogeneous) initial state this structural property is 
preserved during the time evolution. In Section \ref{Subsection3.2} the 
same problem is formulated in terms of GF and its analysis is carried out by
Ovsjannikov-type approximations in a scale of Banach spaces 
(Theorem \ref{Thconv}).

For further applications, let us pointing out that the alternative technical 
standpoint presented in this work shows to be efficient as well on the 
treatment of other types of stochastic dynamics of infinite particle systems, 
namely, the Kawasaki type dynamics in the continuum. This and other cases are 
now being studied and will be reported in forthcoming publications.

\section{General Framework}\label{Section2}

In this section we briefly recall the concepts and results of combinatorial
harmonic analysis on configuration spaces and Bogoliubov generating
functionals needed throughout this work (for a detailed explanation see
\cite{KoKu99,KoKuOl02}).

\subsection{Harmonic analysis on configuration spaces}\label{Subsection2.1}

Let $\Gamma :=\Gamma _{\R^d}$ be the configuration space over $\mathbb{R}^d$,
$d\in\mathbb{N}$,
\[
\Gamma :=\left\{ \gamma \subset \mathbb{R}^d:\left| \gamma\cap\Lambda\right|
<\infty \hbox{
for every compact }\Lambda\subset \mathbb{R}^d\right\} ,
\]
where $\left| \cdot \right|$ denotes the cardinality of a set. We
identify each $\gamma \in \Gamma $ with the non-negative Radon measure
$\sum_{x\in \gamma }\delta_x$ on the Borel $\sigma$-algebra
$\mathcal{B}(\mathbb{R}^d)$, where $\delta_x$ is the Dirac measure with mass
at $x$, which allows to endow $\Gamma$ with the vague topology and the
corresponding Borel $\sigma$-algebra $\mathcal{B}(\Gamma)$.

For any $n\in\N_0:=\N\cup\{0\}$ let
\[
\Gamma^{(n)}:= \{ \gamma\in \Gamma: \vert \gamma\vert = n\},\ n\in \N,\quad \Gamma^{(0)} := \{\emptyset\}.
\]
Clearly, each $\Gamma^{(n)}$, $n\in\N$, can be identify with the
symmetrization of the set $ \{(x_1,...,x_n)\in (\R^d)^n: x_i\not=
x_j \hbox{ if } i\not= j\}$ under the permutation group over
$\{1,...,n\}$, which induces a natural (metrizable) topology on
$\Gamma^{(n)}$ and the corresponding Borel $\sigma$-algebra
$\mathcal{B}(\Gamma^{(n)})$ as well. This leads to the space of
finite configurations
\[
\Gamma_0 := \bigsqcup_{n=0}^\infty \Gamma^{(n)}
\]
endowed with the topology of disjoint union of topological spaces and the
corresponding Borel $\sigma$-algebra $\mathcal{B}(\Gamma_0)$.

Let now $\mathcal{B}_c(\R^d)$ be the set of all bounded Borel sets in
$\R^d$, and for each $\Lambda\in \mathcal{B}_c(\R^d)$ let
$\Gamma_\Lambda := \{\eta\in \Gamma: \eta\subset \Lambda\}$.
Evidently $\Gamma_\Lambda = \bigsqcup_{n=0}^\infty
\Gamma_\Lambda^{(n)}$, where $\Gamma_\Lambda^{(n)}:= \Gamma_\Lambda
\cap \Gamma^{(n)}$, $n\in \N_0$, leading 
to a situation similar to the one for $\Gamma_0$, described above. Given a 
complex-valued $\mathcal{B}(\Gamma_0)$-measurable function $G$ such that
$G\!\!\upharpoonright _{\Gamma\backslash\Gamma_\Lambda}\equiv 0$ for
some $\Lambda \in \mathcal{B}_c(\R^d)$,  the $K$-transform of $G$ is
a mapping $KG:\Gamma\to\mathbb{C}$ defined at each
$\gamma\in\Gamma$ by
\begin{equation}
(KG)(\gamma ):=\sum_{{\eta \subset \gamma}\atop{\vert\eta\vert < \infty} }
G(\eta ).
\label{Eq2.9}
\end{equation}
It has been shown in \cite{KoKu99} that the $K$-transform is a linear and
invertible mapping.

Among the functions in the domain of the $K$-transform we distinguish
the so-called coherent states $e_\lambda(f)$, defined for
complex-valued $\mathcal{B}(\R^d)$-meas\-urable functions $f$ by
$$
e_\lambda (f,\eta ):=\prod_{x\in \eta }f\left( x\right) ,\ \eta \in
\Gamma _0\!\setminus\!\{\emptyset\},\quad  e_\lambda (f,\emptyset ):=1.
$$
The special role of these functions is partially due to the fact that
their image under the $K$-transform coincides with the integrand functions of 
generating functionals (Subsection \ref{Subsection2.2} below). More precisely, 
for any $f$ described as before having in addition compact support, for all 
$\gamma\in \Gamma$
\begin{equation}\label{Kcog}
\left( Ke_\lambda (f)\right) (\gamma )=\prod_{x\in \gamma }(1+f(x)).
\end{equation}

Let $\mathcal{M}_{\mathrm{fm}}^1(\Gamma)$ be the set of all
probability measures $\mu$ on $(\Gamma ,\mathcal{B}(\Gamma))$ with
finite local moments of all orders, i.e.,
\[
\int_\Gamma d\mu (\gamma)\, |\gamma\cap\Lambda |^n<\infty\quad
\mathrm{for\,\,all}\,\,n\in\N
\mathrm{\,\,and\,\,all\,\,} \Lambda \in \mathcal{B}_c(\R^d),
\]
and let $B_{\mathrm{bs}}(\Gamma_0)$ be the set of all complex-valued
bounded $\mathcal{B}(\Gamma_0)$-measurable functions with bounded
support, i.e., $G\!\!\upharpoonright _{\Gamma _0\backslash
\left(\bigsqcup_{n=0}^N\Gamma _\Lambda ^{(n)}\right) }\equiv 0$ for
some $N\in\N_0, \Lambda \in \mathcal{B}_c(\R^d)$. Given a
$\mu\in\mathcal{M}_{\mathrm{fm}}^1(\Gamma)$, the so-called
correlation measure $\rho_\mu$ corresponding to $\mu$ is a measure
on $(\Gamma _0,\mathcal{B}(\Gamma _0))$ defined for all $G\in
B_{\mathrm{bs}}(\Gamma_0)$ by
\begin{equation}
\int_{\Gamma _0}d\rho _\mu(\eta )\,G(\eta )=\int_\Gamma d\mu (\gamma)\, \left(
KG\right) (\gamma).  \label{Eq2.16}
\end{equation}
Observe that under the above conditions $K\!\left|G\right|$ is
$\mu$-integrable. In terms of correlation measures this means that
$B_{\mathrm{bs}}(\Gamma_0)\subset
L^1(\Gamma_0,\rho_\mu)$.\footnote{Throughout this work all
$L^p$-spaces, $p\geq 1$, consist of complex-valued functions.}

Actually, $B_{\mathrm{bs}}(\Gamma_0)$ is dense in
$L^1(\Gamma_0,\rho_\mu)$. Moreover, still by (\ref{Eq2.16}), on
$B_{\mathrm{bs}}(\Gamma_0)$ the inequality $\Vert
KG\Vert_{L^1(\mu)}\leq \Vert G\Vert_{L^1(\rho_\mu)}$ holds, allowing
an extension of the $K$-transform to a bounded operator
$K:L^1(\Gamma_0,\rho_\mu)\to L^1(\Gamma,\mu)$ in such a way that
equality (\ref{Eq2.16}) still holds for any $G\in
L^1(\Gamma_0,\rho_\mu)$. For the extended operator the explicit form
(\ref{Eq2.9}) still holds, now $\mu$-a.e. This means, in particular,
that for all $\mathcal{B}(\R^d)$-measurable functions $f$ such that
$e_\lambda(f)\in L^1(\Gamma_0,\rho_\mu)$ equality \eqref{Kcog} holds for 
$\mu$-a.a.~$\gamma\in\Gamma$ .

\begin{example} The Poisson measure $\pi:=\pi_{dx}$ with intensity the
Lebesgue measure $dx$ on $\R^d$ is the probability measure defined on
$(\Gamma,\mathcal{B}(\Gamma))$ by
\[
\int_\Gamma d\pi(\gamma)\,\exp \left( \sum_{x\in \gamma }\varphi (x)\right)
=\exp \left( \int_{\R^d}dx\,\left( e^{\varphi (x)}-1\right)\right)
\]
for all real-valued smooth functions $\varphi$ on $\R^d$ with
compact support. The correlation measure corresponding to $\pi$ is
the so-called Lebesgue--Poisson measure,
\begin{equation*}
\lambda:=\lambda_{dx}:=\sum_{n=0}^\infty \frac{1}{n!} m^{(n)},
\end{equation*}
where $m^{(n)}$, $n\in \N$, is the measure on
$(\Gamma^{(n)},\mathcal{B}(\Gamma^{(n)}))$ obtained by symmetrization of the
Lebesgue product measure $(dx)^{\otimes n}$ through the symmetrization
procedure described above. For $n=0$ we set $m^{(0)}(\{\emptyset\}):=1$. This 
special case emphasizes the technical role of coherent states in our setting, 
namely, due to the fact $e_\lambda(f)\in L^p(\Gamma_0,\lambda)$ whenever 
$f\in L^p:=L^p(\R^d,dx)$ for some $p\geq 1$, and, moreover,
$\| e_\lambda(f)\|^p_{L^p}=\exp(\| f\|^p_{L^p})$. In particular, for $p=1$, one
additionally has
\begin{equation}
\int_{\Gamma_0}d\lambda(\eta)\, e_\lambda(f,\eta)= \exp\left(\int_{\R^d}dx\,f(x)\right),\label{meanLP}
\end{equation}
for all $f\in L^1$. For more details see \cite{KoKuOl00b}.
\end{example}

\subsection{Bogoliubov generating functionals}\label{Subsection2.2}

Given a probability measure $\mu$ on $(\Gamma, \mathcal{B} (\Gamma))$ the
so-called Bogoliubov generating functional (shortly GF) $B_\mu$
corresponding to $\mu$ is the functional defined at each
$\mathcal{B}(\R^d)$-measurable function $\theta$ by
\begin{equation}
B_\mu(\theta) :=\int_\Gamma d\mu(\gamma)\,\prod_{x\in\gamma}(1+\theta (x)),
\label{Dima2}
\end{equation}
provided the right-hand side exists. It is clear from (\ref{Dima2}) 
that one cannot define the GF for all probability measures on $\Gamma$ but, if 
it exists for some measure $\mu$, then the domain of $B_\mu$ depends on 
$\mu$ and, conversely, the domain of $B_\mu$ reflects special properties over 
the underlying measure $\mu$ \cite{KoKuOl02}. For instance, if $\mu$ has finite 
local exponential moments, i.e.,
\[
\int_\Gamma d\mu (\gamma )\, e^{\alpha|\gamma\cap\Lambda |}<\infty \quad
\hbox{for all}\,\,\alpha>0\,\,
\hbox{and all}\,\,\Lambda \in \mathcal{B}_c(\R^d),
\]
then $B_\mu$ is well-defined, for instance, on all bounded functions
$\theta$ with compact support. According to the previous subsection, this 
implies that to such a measure $\mu$ one may associate the correlation measure 
$\rho_\mu$, leading to a description of the functional $B_\mu$ in terms of 
either the measure $\rho_\mu$:
\[
B_\mu(\theta)
= \int_\Gamma d\mu(\gamma)\,\left( Ke_\lambda (\theta)\right) (\gamma)
= \int_{\Gamma_0}d\rho_\mu(\eta)\, e_\lambda (\theta, \eta),
\]
or the so-called correlation function $k_\mu:=\frac{d\rho_\mu}{d\lambda}$
corresponding to the measure $\mu$, if $\rho_\mu$ is absolutely continuous
with respect to the Lebesgue--Poisson measure $\lambda$:
\begin{equation}\label{BF_via_cf}
B_\mu(\theta)=\int_{\Gamma_0}d\lambda(\eta)\, e_\lambda (\theta, \eta)k_\mu(\eta).
\end{equation}

Throughout this work we will consider GF defined on the whole complex $L^1$ 
space. Furthermore, we will assume that the GF are entire. We recall that a 
functional $A:L^1\to\C$ is entire on $L^1$ whenever $A$ is locally bounded and 
for all $\theta_0,\theta\in L^1$ the mapping 
$\C\ni z\mapsto A(\theta_0 + z\theta)\in\C$ is entire. Thus, at each 
$\theta_0\in L^1$, every entire functional $A$ on $L^1$ has a representation 
in terms of its Taylor expansion,
$$
A(\theta_0+ z \theta)=\sum_{n=0}^\infty \frac{z^n}{n!}
d^nA(\theta_0;\theta ,...,\theta),\quad z\in\C, \theta\in L^1.
$$

The next theorem states properties specific for entire functionals $A$ on
$L^1$ and their higher order derivatives $d^nA(\theta_0;\cdot)$ (for a
detailed explanation see \cite{KoKuOl02} and the references therein).

\begin{theorem}
\label{9Prop9.1.1}Let $A$ be an entire functional on $L^1$. Then
each differential $d^nA(\theta _0;\cdot), n\in\N, \theta _0\in L^1$
is defined by a symmetric kernel $$\vd^n A(\theta_0;\cdot)\in
L^\infty (\R^{dn}):=L^\infty \bigl((\R^d)^n,(dx)^{\otimes n}\bigr)$$
called the variational derivative of $n$-th order of $A$ at the
point $\theta _0$. More precisely, 
\begin{eqnarray*}
d^nA(\theta _0;\theta _1,...,\theta _n) &:=&\frac{\partial ^n}{\partial
z_1...\partial z_n}A\left( \theta _0+\sum_{i=1}^nz_i\theta _i\right) %
\Bigg\vert_{z_1=...=z_n=0} \\
&\hbox{\rm{=:}}&\int_{(\R^d)^n}dx_1\ldots dx_n\,\vd^n
A(\theta_0;x_1,\ldots,x_n)\prod_{i=1}^n\theta_i(x_i)
\end{eqnarray*}
for all $\theta _1,...,\theta _n\in L^1$. Moreover, the operator
norm of the bounded $n$-linear functional $d^nA(\theta_0;\cdot)$ is
equal to $\left\| \vd^n
A(\theta_0;\cdot)\right\|_{L^\infty(\R^{dn})}$ and for all $r>0$ one
has
\begin{equation}
\left\| \vd A (\theta_0;\cdot)\right\|_{L^\infty(\R^d)} \leq
\frac{1}{r} \sup_{\|\theta^\prime \|_{L^1} \leq r}
|A(\theta_0+\theta^\prime)|\label{2011!}
\end{equation}
and, for $n\geq 2$,
\begin{equation}
\left\| \vd^nA(\theta_0;\cdot)\right\|_{L^\infty(\R^{dn})} \leq n!
\left(\frac{e}{r}\right)^n \sup_{\|\theta^\prime \|_{L^1} \leq r}
|A(\theta_0+\theta^\prime)|. \label{Duarte}
\end{equation}
\end{theorem}

The first part of Theorem \ref{9Prop9.1.1} stated for GF and their variational
derivatives at $\theta_0=0$ yields the next result.

\begin{proposition}
\label{9Prop9.1.3} Let $B_\mu$ be an entire GF on
$L^1$. Then the measure $\rho_\mu$ is absolutely continuous
with respect to the Lebesgue--Poisson measure $\lambda$ and
the Radon--Nykodim derivative $k_\mu=\dfrac{d\rho_\mu}{d\lambda}$ is
given by
\[
k_\mu(\eta )= \vd^{\left|\eta\right|} B_\mu(0;\eta)
\quad \text{for $\lambda$-a.a.
}\eta \in \Gamma _0. 
\]
\end{proposition}

Concerning the second part of Theorem \ref{9Prop9.1.1}, namely, estimates 
(\ref{2011!}) and (\ref{Duarte}), we note that $A$ being entire does not ensure 
that for every $r>0$ the supremum appearing on the right-hand side of 
(\ref{2011!}), (\ref{Duarte}) is always finite. This will hold if, in 
addition, the entire functional $A$ is of bounded type, that is,
\[
\forall\,r>0,\
\sup_{\Vert\theta\Vert_{L^1}\leq r}\left|A(\theta_0+\theta)\right|<\infty,\ \forall\,\theta_0\in L^1.
\]
Hence, as a consequence of Proposition \ref{9Prop9.1.3}, it follows from
(\ref{2011!}) and (\ref{Duarte}) that the correlation function $k_\mu$ of an 
entire GF of bounded type on $L^1$ fulfills the so-called generalized Ruelle
bound, that is, for any $0\leq\varepsilon\leq 1$ and any $r>0$ there is some 
constant $C\geq 0$ depending on $r$ such that
\begin{equation}
k_\mu(\eta )\leq C\left(\left|\eta\right|!\right)^{1-\varepsilon}
\left( \frac er\right) ^{\left|\eta\right|},\quad
\lambda \mathrm{-a.a.}\text{ }\eta \in\Gamma _0.\label{KuRu}
\end{equation}
In our case, $\varepsilon=0$. We observe that if (\ref{KuRu}) holds for 
$\varepsilon =1$ and for at least one $r>0$, then condition (\ref{KuRu}) is 
the classical Ruelle bound. In terms of GF, the latter means that 
\[
|B_\mu(\theta)|\leq C\exp\left(\dfrac{e}{r}\|\theta\|_{L^1}\right),
\]
as can be easily checked using the representation \eqref{BF_via_cf} and 
\eqref{meanLP}. This special case motivates the definition of the following 
family of Banach spaces, see \cite[Proposition 23]{KoKuOl02}.

\begin{definition}\label{Bs} For each $\alpha>0$, let $\mathcal{E}_\alpha$ be 
the Banach space of all entire functionals $B$ on $L^1$ such that
\[
\left\| B\right\| _\alpha :=\sup_{\theta \in L^1}
\left( \left|
B(\theta )\right| e^{-\frac{1}{\alpha} \left\| \theta \right\| _{L^1}}\right)
<\infty .
\]
\end{definition}

\subsection{Time evolution equations}\label{Subsection2.3}

Informally, the stochastic evolution of an interacting particle 
system on $\mathbb{R}^d$ may be described through a Markov generator $L$ 
defined on a proper space of functions on $\Gamma$. The problem of 
construction of the corresponding Markov process on $\Gamma$
is related to the existence (on a proper space of functions) of the 
semigroup corresponding to $L$, which will be the solution to a Cauchy 
problem
\begin{equation*}
\frac{dF_t}{dt}=LF_t,\qquad F_t\!\bigm|_{t=0}=F_0.
\end{equation*}
However, from the technical point of view, to show that $L$ is the generator of 
a semigroup on some reasonable space of functions defined on $\Gamma$ seems to 
be often a difficult question.

In applications, the properties of the evolution of the system through its 
states, that is, probability measures on $\Gamma$, is a subject of interest. 
Informally, such a time evolution is given by the dual Kolmogorov equation, 
the so-called Fokker-Planck equation,
\begin{equation}
\frac{d\mu_t}{dt}=L^*\mu_t, \qquad
\mu_t\!\bigm|_{t=0}=\mu_0\label{FokkerPlanck},
\end{equation}
where $L^*$ is the dual operator of $L$. Technically, the use of definition 
\eqref{Eq2.16} allows an alternative approach to the study of 
\eqref{FokkerPlanck} through the corresponding correlation functions 
$k_t:=k_{\mu_t}$, $t\geq0$, provided they exist. This leads to the Cauchy 
problem
\begin{equation*}
\frac \partial {\partial t}k_t=\hat L^*k_t,\quad
{k_t}_{|t=0}=k_0,
\end{equation*}
where $k_0$ is the correlation function corresponding to the initial 
distribution $\mu_0$ and $\hat L^*$ is the dual operator of $\hat L:=K^{-1}LK$ 
in the sense
\begin{equation*}
\int_{\Gamma_0}d\lambda(\eta)\,(\hat LG)(\eta) k(\eta)=
\int_{\Gamma_0}d\lambda(\eta)\,G(\eta) (\hat L^*k)(\eta).
\end{equation*}
Through the representation \eqref{BF_via_cf}, this gives us a way to express 
the dynamics also in terms of the GF $B_t$ corresponding to $\mu_t$, i.e., 
informally,
\begin{align}
\frac \partial {\partial t}B_t(\theta ) &=\int_{\Gamma _0}d\lambda(\eta)\,e_\lambda(\theta ,\eta )\left( \frac \partial{\partial t}k _t(\eta )\right)=\int_{\Gamma _0}d\lambda(\eta)\, e_\lambda(\theta ,\eta )(\hat L^*k_t)(\eta )\nonumber\\
&=\int_{\Gamma _0}d\lambda(\eta)\,(\hat Le_\lambda(\theta))(\eta)k_t(\eta)=:(\tilde L B_t)(\theta).\label{obtevol}
\end{align}

Concerning the evolution equation
\begin{equation}
\frac{\partial B_t}{\partial t}=\tilde LB_t,\label{evolgf}
\end{equation}
we observe that from the previous construction follows that if a solution 
$B_t$, $t\geq 0$, exists for some GF as an initial condition, then one may 
expect that each $B_t$ is the GF corresponding to the state of 
the system at the time $t$. However, besides the existence problem, at this 
point it is opportune to underline that if a solution to (\ref{evolgf}) 
exists, a priori it does not have to be a GF (corresponding to some measure). 
This verification requests an additional analysis, see e.g.~\cite{KoKuOl02}, 
\cite{K00}. 

In most concrete applications, to find a solution to (\ref{evolgf}) on a 
Banach space seems to be often a difficult question. However, the problem may 
be simplified within the framework of scales of Banach spaces. We recall that 
a scale of Banach spaces is a one-parameter family of Banach spaces 
$\{\B_s: 0<s\leq s_{0}\}$ such that
\[
\B_{s''}\subseteq \B_{s'},\quad \|\cdot\|_{s'}\leq \|\cdot\|_{s''}
\]
for any pair $s'$, $s''$ such that $0<s'< s''\leq s_{0}$, where $\|\cdot\|_s$ 
denotes the norm in $\B_s$. As an example, it is clear from Definition 
\ref{Bs} that for any $\alpha_0>0$ the family 
$\{\mathcal{E}_\alpha: 0<\alpha\leq\alpha_0\}$ is a scale of Banach spaces.

Within this framework, one has the following existence and uniqueness result
(see e.g.~\cite{T68}). 

\begin{theorem}\label{Th1} On a scale of Banach spaces $\{\B_s: 0<s\leq s_0\}$ 
consider the initial value problem
\begin{equation}
\frac{du(t)}{dt}=Au(t),\quad u(0)=u_0\in
\mathbb{B}_{s_0}\label{V1}
\end{equation}
where, for each $s\in(0,s_0)$ fixed and for each pair $s', s''$ such that 
$s\leq s'<s''\leq s_0$, $A:\B_{s''}\to\B_{s'}$ is a linear mapping so that 
there is an $M>0$ such that for all $u\in\B_{s''}$
\[
\|Au\|_{s'}\leq\frac{M}{s''-s'}\|u\|_{s''}.
\]
Here $M$ is independent of $s',s''$ and $u$, however it might depend 
continuously on $s,s_0$.

Then, for each $s\in(0,s_0)$, there is a constant $\delta>0$ (which depends on 
$M$) such that there is a unique function 
$u:\bigl[0,\delta(s_0-s)\bigr)\rightarrow\mathbb{B}_s$ which is continuously
differentiable on $\bigl(0,\delta(s_0-s)\bigr)$ in $\mathbb{B}_s$,
$Au\in\mathbb{B}_s$, and solves \eqref{V1} in the time-interval
$0\leq t<\delta(s_0-s)$.
\end{theorem}

In Appendix we present a sketch of the proof of Theorem \ref{Th1}, which will 
be used to prove Theorem \ref{Thconv} below.

\section{The Glauber dynamics}\label{Subsection3.1}

The Glauber dynamics is an example of a birth-and-death model where, in this
special case, particles appear and disappear according to a death rate
identically equal to 1 and to a birth rate depending on the interaction
between particles. More precisely, let $\phi:\R^d\to\R\cup\{+\infty\}$ be a 
pair potential, that is, a $\mathcal{B}(\R^d)$-measurable function such that 
$\phi(-x)=\phi(x)\in \R$ for all $x\in \R^d\setminus\{0\}$, which we will 
assume to be non-negative and integrable. Given a configuration $\gamma\in\Gamma$, the
birth rate of a new particle at a site $x\in\R^d\setminus\gamma$ is given
by $\exp(-E(x,\gamma))$, where $E(x,\gamma)$ is a relative energy of
interaction between a particle located at $x$ and the configuration $\gamma$
defined by
\[
E(x,\gamma ):=\sum_{y\in \gamma }\phi (x-y)\in[0,+\infty].
\]
Informally, in terms of Markov generators, this means that the behavior of 
such an infinite particle system is described by
\begin{align}
(LF)(\gamma):=&\sum_{x\in \gamma}\big(F(\gamma\setminus\{x\}) - F(\gamma)\big)\nonumber
\\&+ z\int_{\R^d} dx\,e^{-E(x,\gamma)} \big(F(\gamma\cup\{x\}) - F(\gamma)\big),\label{Vl41}
\end{align}
where $z>0$ is an activity parameter (for more details see 
e.g.~\cite{FKO05,KoKtZh06}). As a consequence of Subsection 
\ref{Subsection2.3}, this implies that the operator $\tilde L$ defined in 
(\ref{obtevol}) is given cf.~\cite{FKO05} by
\begin{equation}\label{LtildeGL}
(\tilde LB)(\theta)= -\int_{\R^d} dx\,\theta(x)\Bigl( \vd
B(\theta;x) - zB\left(\theta e^{-\phi (x -\cdot )}+e^{-\phi (x
-\cdot )}-1\right)\Bigr).
\end{equation}

\begin{theorem}\label{Th13} Given an $\alpha_0>0$, let 
$B_0\in\mathcal{E}_{\alpha_0}$. For each $\alpha\in(0,\alpha_0)$ there is a 
$T>0$ (which depends on $\alpha,\alpha_0$) such that there is a 
unique solution $B_t$, $t\in[0,T)$, to the initial value problem 
$\dfrac{\partial B_t}{\partial t}=\tilde LB_t$, ${B_t}_{|t=0}= B_0$ in the 
space $\mathcal{E}_\alpha$. 
\end{theorem}

This theorem follows as a concrete application of Theorem \ref{Th1} and the 
following estimate of norms.

\begin{proposition}\label{Proposition2}
Let $\alpha_0>\alpha>0$ be given. If $B\in\mathcal{E}_{\alpha''}$ for some 
$\alpha''\in \left(s,s_0\right]$, then $\tilde LB\in\mathcal{E}_{\alpha'}$ for 
all $\alpha\leq\alpha'<\alpha''$, and we have 
\[
\|\tilde LB\|_{\alpha'}\leq \frac{\alpha_{0}}{\alpha''-\alpha'}
\left(1+z\alpha_0 e^{\frac{\|\phi\|_{L^1}}{\alpha}-1}\right)\|B\|_{\alpha''}.
\]
\end{proposition}  

To prove this result, the next two lemmata show to be useful.

\begin{lemma}\label{Lemma1} Given an $\alpha>0$, for all 
$B\in\mathcal{E}_\alpha$ let
\[
(L_0 B)(\theta):=\int_{\R^d}dx\,\theta(x)\vd B(\theta; x),\quad
\theta\in L^1.
\]
Then, for all $\alpha'<\alpha$, we have $L_0B\in\mathcal{E}_{\alpha'}$ and,
moreover, the following estimate of norms holds:
\[
\|L_0B\|_{\alpha'}\leq \frac{\alpha'}{\alpha-\alpha'}\|B\|_{\alpha}.
\]
\end{lemma}
\begin{proof}
First we observe that an application of Theorem \ref{9Prop9.1.1} shows that 
$L_0B$ is an entire functional on $L^1$ and, in addition, that for all $r>0$ 
and all $\theta\in L^1$, 
\[
\left|(L_0B)(\theta)\right|\leq\|\theta\|_{L^1}\left\|\vd
B(\theta;\cdot)\right\|_{L^\infty(\R^d)}
\leq\frac{\|\theta\|_{L^1}}{r}\sup_{\|\theta_0\|_{L^1}\leq
r}\left|B(\theta+\theta_0)\right|,
\]
where, for all $\theta_0\in L^1$ such that $\|\theta_0\|_{L^1}\leq r$,
\[
\left|B(\theta+\theta_0)\right|\leq \|B\|_\alpha e^{\frac{\|\theta\|_{L^1}}{\alpha}+\frac{r}{\alpha}}.
\]
Thus,
\[
\|L_0B\|_{\alpha'}=\sup_{\theta\in L^1}\left(e^{-\frac{1}{\alpha'}\|\theta\|_{L^1}}|(L_0B)(\theta)|\right)\leq\frac{e^{\frac{r}{\alpha}}}{r}\|B\|_\alpha\sup_{\theta\in L^1}\left(e^{-\left(\frac{1}{\alpha'}-\frac{1}{\alpha}\right)\|\theta\|_{L^1}}\|\theta\|_{L^1}\right),
\]
where the latter supremum is finite provided 
$\frac{1}{\alpha'}-\frac{1}{\alpha}>0$. In such a situation, the use of 
inequality $xe^{-mx}\leq\frac{1}{em}$, $x\geq0$, $m>0$ leads for each $r>0$ to
\begin{equation*}
\|L_0B\|_{\alpha'}\leq\frac{e^{\frac{r}{\alpha}}}{r}\frac{\alpha\alpha'}{e(\alpha-\alpha')}\|B\|_\alpha.
\end{equation*}
The required estimate of norms follows by minimizing the expression 
$\frac{e^{\frac{r}{\alpha}}}{r}\frac{\alpha\alpha'}{e(\alpha-\alpha')}$ in the 
parameter $r$, that is, $r=\alpha$.
\end{proof}

\begin{lemma}\label{Lemma2} Let $\varphi,\psi:\R^d\times\R^d\to\R$ be such that,
for a.a.~$x\in\R^d$, $\varphi(x,\cdot)\in L^\infty:=L^\infty(\R^d)$,
$\psi(x,\cdot)\in L^1$ and $\|\varphi(x,\cdot)\|_{L^\infty}\leq c_0$,
$\|\psi(x,\cdot)\|_{L^1}\leq c_1$ for some constants $c_0,c_1>0$ independent
of $x$. For each $\alpha>0$ and all $B\in\mathcal{E}_\alpha$ let
\[
(L_1B)(\theta):=\int_{\R^d}dx\,\theta(x)B(\varphi(x, \cdot)\theta+\psi(x,\cdot)),\quad\theta\in L^1.
\]
Then, for all $\alpha'>0$ such that $c_0\alpha'<\alpha$, we have
$L_1B\in\mathcal{E}_{\alpha'}$
and
\[
\|L_1B\|_{\alpha'}\leq \frac{\alpha\alpha'}{\alpha-c_0\alpha'}e^{\frac{c_1}{\alpha}-1}\|B\|_{\alpha}.
\]
\end{lemma}

\begin{proof}
As before, it follows from Theorem \ref{9Prop9.1.1} that $L_1B$ is an entire 
functional on $L^1$. Hence, given a $B\in\mathcal{E}_\alpha$, for all 
$\theta\in L^1$ one has
\[
|B(\varphi(x,\cdot)\theta+\psi(x,\cdot))|\leq \|B\|_{\alpha}\,e^{\frac{1}{\alpha}\left(\|\varphi(x,\cdot)\theta\|_{L^1}+\|\psi(x,\cdot)\|_{L^1}\right)},
\]
and thus
\begin{eqnarray*}
\|L_1B\|_{\alpha'}&\leq&\sup_{\theta\in L^1}\left(e^{-\frac{1}{\alpha'}\|\theta\|_{L^1}}\int_{\R^d}dx\,\left|\theta(x)B(\varphi(x,\cdot)\theta+\psi(x,\cdot))\right|\right)\\
&\leq& e^{\frac{c_1}{\alpha}}\|B\|_{\alpha}\sup_{\theta\in L^1}\left(e^{-\left(\frac{1}{\alpha'}-\frac{c_0}{\alpha}\right)\|\theta\|_{L^1}}\|\theta\|_{L^1}\right).
\end{eqnarray*}
The proof follows as in the proof of Lemma \ref{Lemma1}.
\end{proof}

\begin{proof}[Proof of Proposition \ref{Proposition2}]
In Lemma \ref{Lemma2} replace $\varphi$ by $e^{-\phi}$ and
$\psi$ by $e^{-\phi}-1$. Due to the positiveness and integrability
properties of $\phi$ one has $e^{-\phi}\leq 1$ and
$|e^{-\phi}-1|=1-e^{-\phi}\leq \phi\in L^1$, ensuring the conditions
to apply Lemma \ref{Lemma2}. This together with Lemma \ref{Lemma1}
leads to the required result.
\end{proof}

\begin{remark}
 It follows from the proof of Theorem \ref{Th1} that for each $t\in(0,T)$
there is an $\alpha_t\in(\alpha,\alpha_0)$ such that
$B_t\in\mathcal{E}_\beta$ for all $\beta\in\left[\alpha,\alpha_t\right)$, 
cf.~\cite{FKKoz2011}.
\end{remark}

\begin{remark}
Concerning the initial conditions considered in Theorem \ref{Th13}, observe 
that, in particular, $B_0$ can be an entire GF $B_{\mu_0}$ on $L^1$ such that, 
for some constants $\alpha_0,C>0$, 
$|B_{\mu_0}(\theta)|\leq C\exp(\frac{\|\theta\|_{L^1}}{\alpha_0})$ for all 
$\theta\in L^1$. As we have mentioned before, in such a situation an 
additional analysis is required in order to guarantee that for each time $t$ 
the local solution $B_t$ given by Theorem \ref{Th13} is a GF. Such an analysis 
is outside of our goal, but it may be done using e.g.~\cite[Theorem 2.13]{GK2006}, 
which yields the existence of a proper $\mathcal{S}\subset\Gamma$ and a 
$\mathcal{S}$-valued process with 
sample paths in the Skorokhod space $D_\mathcal{S}(\left[0,+\infty\right))$ 
associated with $L$ defined in \eqref{Vl41}. This shows the existence of the time 
evolution $\mu_0\mapsto\mu_t$, where $\mu_t$ is the law of the 
$\mathcal{S}$-valued process, leading apart from existence problems to the time evolution 
$B_{\mu_0}\mapsto B_{\mu_t}$ of the corresponding GF. The latter will be a 
solution to the initial value problem \eqref{evolgf}, \eqref{LtildeGL} with 
${B_t}_{|t=0}= B_{\mu_0}$. By the uniqueness stated in Theorem \ref{Th13}, this 
implies that for each $t\in[0, T)$ we will have $B_t=B_{\mu_t}$, and thus 
$B_t$ is a GF.
\end{remark}

Theorem \ref{Th13} only ensures the existence of a local solution. However, 
under certain initial conditions, such a solution might be extend to a global 
one, that is, to a solution defined on the whole time interval 
$\left[0,+\infty\right)$, as follows. Assume that the initial condition $B_0$ 
is an entire GF on $L^1$. Then, by Proposition \ref{9Prop9.1.3}, $B_0$ can be 
written in terms of the corresponding correlation function $k_0$, 
\[
B_0(\theta)=\int_{\Gamma_0}d\lambda(\eta)\, e_\lambda (\theta, \eta)k_0(\eta),\quad \theta\in L^1.
\]
Assuming, in addition, that $k_0$ fulfills the Ruelle bound 
$k_0(\eta)\leq z^{|\eta|}$, $\eta\in\Gamma_0$, being $z$ the activity parameter 
appearing in definition (\ref{Vl41}), then, in terms of $B_0$, this leads to
$|B_0(\theta)|\leq e^{z\|\theta\|_{L^1}}$, $\theta\in L^1$, showing that 
$B_0\in\mathcal{E}_{1/z}$ and, moreover, $\|B_0\|_{\frac{1}{z}}\leq 1$. Thus, 
fixing an $\alpha\in\left(0,1/z\right)$, an application of 
Theorem~\ref{Th13} yields a solution $B_t$, 
$t\in\left[0,\delta(1/z-\alpha\right))$, to the initial value problem 
\eqref{evolgf}, \eqref{LtildeGL} with ${B_t}_{|t=t_0}= B_0$. Assume that each 
$B_t$ is an entire GF on $L^1$. As shown in \cite[Lemma 3.10]{FKKoz2011}, in 
this case the corresponding correlation function $k_t$ still fulfills the 
Ruelle bound with the same constant $z$, meaning that the local solution  
$B_t$, $t\in\left[0,\delta(1/z-\alpha\right))$, does not leave the initial 
Banach space $\mathcal{E}_{1/z}$. This allows us to consider any 
$t_0\in\left[0,\delta(1/z-\alpha)\right)$ sufficiently close to 
$\delta(1/z-\alpha)$ as an initial time and, as before, to study the 
initial value problem \eqref{evolgf}, \eqref{LtildeGL} with 
${B_t}_{|t=t_0}= B_{t_0}$ in the same scale of Banach spaces.  This will give a 
solution $B_t$ on the time-interval $\left[t_0,t_0+\delta(1/z-\alpha)\right)$. 
Assuming again that each $B_t$, $t\in\left[t_0,t_0+\delta(1/z-\alpha)\right)$, 
is an entire GF on $L^1$, naturally that $B_t\in\mathcal{E}_{1/z}$ with 
$\|B_t\|_{1/z}\leq 1$, for all $t\in\left[t_0,t_0+\delta(1/z-\alpha)\right)$. 
Therefore, one may repeat the above arguments. 

This argument iterated yields at the end a solution to the initial value 
problem $\frac{\partial B_t}{\partial t}=\tilde LB_t$, ${B_t}_{|t=0}= B_0$ 
defined on $\left[0,+\infty\right)$. Of course, by the uniqueness stated in 
Theorem \ref{Th13}, the global solution constructed in this way is necessarily 
unique. In this context, one may state the following result.

\begin{corollary}\label{bdd}
Given an entire GF $B_0$ on $L^1$ such that the corresponding correlation 
function $k_0$ fulfills the Ruelle bound $k_0(\eta)\leq z^{|\eta|}$, 
$\eta\in\Gamma_0$, for the activity parameter $z$ appearing in definition 
(\ref{Vl41}), the local solution to the initial value problem 
$\dfrac{\partial B_t}{\partial t}=\tilde LB_t$, ${B_t}_{|t=0}= B_0$ (given by 
Theorem \ref{Th13}) might be extended to a global solution which, for each time 
$t\geq 0$, is an entire GF on $L^1$. 
\end{corollary}

\section{Vlasov scaling}\label{Subsection3.2}

We proceed to investigate the Vlasov-type scaling proposed in \cite{FKK10a} 
for generic continuous particle systems and accomplished in \cite{FKK10} for 
the Glauber dynamics, now in terms of GF. As explained in both references, the 
aim is to construct a scaling of the operator $L$ defined in \eqref{Vl41}, 
$L_\eps$, $\eps >0$, in such a way that the rescale of a starting 
correlation function $k_0$, denote by $k_0^{(\eps)}$, provides a singularity 
with respect to $\eps$ of the type 
$k_0^{(\eps)}(\eta) \sim \eps^{-|\eta|} r_0(\eta)$, $\eta\in\Gamma_0$, being 
$r_0$ a function independent of $\eps$, and, moreover, the following two 
conditions are fulfilled. The first one is that under the scaling 
$L\mapsto L^\eps$ the solution $k^{(\varepsilon)}_t$, $t\geq 0$, to
\begin{equation*}
\frac \partial {\partial t}k_t^{(\varepsilon)}=\hat L_\varepsilon^*k^{\varepsilon}_t,\quad {k^{(\varepsilon)}_t}_{|t=0}=k_0^{(\varepsilon)}
\end{equation*}   
preserves the order of the singularity with respect to $\eps$, that is, 
$k_t^{(\eps)}(\eta) \sim \eps^{-|\eta|} r_t(\eta)$, $\eta\in\Gamma_0$. The 
second condition is that the dynamics $r_0 \mapsto r_t$ preserves the 
Lebesgue-Poisson exponents, that is, if $r_0$ is of the form 
$r_0=e_\lambda(\rho_0)$, then each $r_t$, $t>0$, is of the same type, i.e.,
$r_t=e_\lambda(\rho_t)$, where $\rho_t$ is a solution to a non-linear equation
(called a Vlasov-type equation). As shown in \cite[Example 8]{FKK10a}, 
\cite{FKK10}, this equation is given by
\begin{equation}
\frac{\partial}{\partial t}\rho_t(x)=-\rho_t(x)+ze^{-(\rho_t*\phi)(x)},\quad x\in\R^d,\label{Vl42}
\end{equation}
where $*$ denotes the usual convolution of functions. Existence of classical
solutions $0\leq\rho_t\in L^\infty$ to \eqref{Vl42} has been discussed in 
\cite{FKK10}, \cite{FKKoz2011}. Therefore, it is natural to consider the same 
scaling, but in GF.

The previous scheme was accomplished in \cite{FKK10} through the scale 
transformations $z\mapsto \varepsilon^{-1}z$ and $\phi\mapsto\varepsilon\phi$ 
of the operator $L$, that is,
\[
(L_\varepsilon F)(\gamma):=\sum_{x\in \gamma}\big(F(\gamma\setminus\{x\}) - F(\gamma)\big) 
+ \frac{z}{\varepsilon}\int_{\R^d} dx\,e^{-\varepsilon E(x,\gamma)} \big(F(\gamma\cup\{x\}) - F(\gamma)\big).
\]

To proceed towards GF, we consider $k^{(\varepsilon)}_t$ defined as before and 
$k^{(\varepsilon)}_{t,\mathrm{ren}}(\eta):=\varepsilon^{|\eta|}k^{(\varepsilon)}_t(\eta)$. In terms of GF, these yield
\[
B_t^{(\eps)}(\theta):=\int_{\Gamma_0} d\lambda(\eta)
e_\lambda(\theta,\eta)k_t^{(\eps)}(\eta),
\]
and
\[
B_{t,\mathrm{ren}}^{(\varepsilon)}(\theta):=\int_{\Gamma_0}d\lambda(\eta)\,e_\lambda(\theta,\eta)k_{t,\mathrm{ren}}^{(\varepsilon)}(\eta)=\int_{\Gamma_0}d\lambda(\eta)\,e_\lambda(\varepsilon\theta,\eta)k_t^{(\varepsilon)}(\eta)=B_t^{(\varepsilon)}(\varepsilon\theta),
\]
leading, as in \eqref{obtevol}, to the initial value problem
\begin{equation}
\frac{\partial}{\partial t}B_{t,\mathrm{ren}}^{(\varepsilon)}
=\tilde L_{\varepsilon, \mathrm{ren}}B_{t,\mathrm{ren}}^{(\varepsilon)},\quad
{B^{(\varepsilon)}_{t,\mathrm{ren}}}_{|t=0}=B_{0,\mathrm{ren}}^{(\varepsilon)}.\label{V12}
\end{equation}

\begin{proposition}
For all $\varepsilon>0$ and all $\theta\in L^1$, we have
\[
(\tilde L_{\varepsilon, \mathrm{ren}}B)(\theta)=
-\int_{\R^d}dx\,\theta(x)\left( \vd B(\theta,x) - zB\Bigl(\theta
e^{-\varepsilon\phi (x -\cdot )}+\frac{e^{-\varepsilon\phi (x -\cdot
)}-1}{\varepsilon}\Bigr)\right).
\]
\end{proposition}
\begin{proof}
As shown in \cite[Proposition 3.1]{FKK10},
\begin{eqnarray*}
&&(\hat L_{\varepsilon, \mathrm{ren}}e_\lambda(\theta))(\eta)\\
&=&-|\eta|e_\lambda(\theta,\eta)+z\sum_{\xi\subseteq\eta}e_\lambda(\theta,\xi)\int_{\R^d}dx\,\theta(x)e^{-\varepsilon E(x,\xi)}
e_\lambda\left(\frac{e^{-\varepsilon\phi(x-\cdot)}-1}{\varepsilon},\eta\setminus\xi\right).
\end{eqnarray*}
Therefore,
\[
(\tilde L_{\varepsilon, \mathrm{ren}}B)(\theta)=
\int_{\Gamma_0}d\lambda(\eta)\,(\hat L_{\varepsilon, \mathrm{ren}}e_\lambda(\theta))(\eta)k(\eta),
\]
with
\[
\int_{\Gamma_0}d\lambda(\eta)\,|\eta|e_\lambda(\theta,\eta)k(\eta)=
\int_{\R^d}dx\,\theta(x)\vd B(\theta;x)
\]
and
\begin{eqnarray*}
&&\int_{\Gamma_0}d\lambda(\eta)\,k(\eta)\sum_{\xi\subseteq\eta}e_\lambda(\theta,\xi)\int_{\R^d}dx\,\theta(x)e^{-\varepsilon E(x,\xi)}
e_\lambda\left(\frac{e^{-\varepsilon\phi(x-\cdot)}-1}{\varepsilon},\eta\setminus\xi\right)\\
&=&\int_{\R^d}dx\,\theta(x)\int_{\Gamma_0}d\lambda(\eta)\,e_\lambda\left(\frac{e^{-\varepsilon\phi(x-\cdot)}-1}{\varepsilon},\eta\right)\int_{\Gamma_0}d\lambda(\xi)\,
k(\eta\cup\xi)e_\lambda(\theta e^{-\varepsilon\phi(x-\cdot)},\xi)\\
&=&\int_{\R^d}dx\,\theta(x)\int_{\Gamma_0}d\lambda(\eta)\,e_\lambda\left(\frac{e^{-\varepsilon\phi(x-\cdot)}-1}{\varepsilon},\eta\right)
\vd^{|\eta|}B (\theta e^{-\varepsilon\phi(x-\cdot)};\eta)\\
&=&\int_{\R^d}dx\,\theta(x)B\left(\theta e^{-\varepsilon\phi (x -\cdot )}+\frac{e^{-\varepsilon\phi (x -\cdot )}-1}{\varepsilon}\right),
\end{eqnarray*}
where we have used the relation between variational derivatives derived in 
\cite[Proposition 11]{KoKuOl02}.
\end{proof}

\begin{proposition}\label{estimativas} (i) If $B\in\mathcal{E}_{\alpha}$ for 
some $\alpha >0$, then, for all $\theta\in L^1$, 
$(\tilde L_{\varepsilon,\mathrm{ren}}B)(\theta)$ converges as $\eps$ tends zero 
to
\[
(\tilde L_V B)(\theta):= -\int_{\R^d}dx\,\theta(x)\left( \vd
B(\theta;x) - zB\left(\theta
-\phi(x-\cdot)\right)\right).
\]
(ii) Let $\alpha_0>\alpha>0$ be given. If $B\in\mathcal{E}_{\alpha''}$ for some 
$\alpha''\in (\alpha, \alpha_0]$, then $\bigl\{\tilde L_{\varepsilon, \mathrm{ren}}B,\tilde{L}_V B\bigr\} \subset\mathcal{E}_{\alpha'}$ for all 
$\alpha\leq\alpha'<\alpha''$, and we have 
\[
\|\tilde{L}_{\#}B\|_{\alpha'}\leq \frac{\alpha_{0}}{\alpha''-\alpha'}
\left(1+z\alpha_0
e^{\frac{\|\phi\|_{L^1}}{\alpha}-1}\right)\|B\|_{\alpha''},
\]
where $\tilde{L}_{\#}=\tilde{L}_{\varepsilon, \mathrm{ren}}$ or 
$\tilde{L}_{\#}=\tilde{L}_V$.
\end{proposition}

\begin{proof} (i) First we observe that for 
a.a.~$x\in\R^d$ one clearly has
\[
\lim_{\varepsilon\searrow 0}\left(\theta e^{-\varepsilon\phi (x -\cdot )}+\frac{e^{-\varepsilon\phi (x -\cdot )}-1}{\varepsilon}\right)=\theta -\phi(x-\cdot)\ \mbox{in}\ L^1, 
\]
and thus, due to the continuity of $B$ in $L^1$ ($B$ is even entire on $L^1$), 
the following limit holds
\[
\lim_{\varepsilon\searrow 0}B\left(\theta e^{-\varepsilon\phi (x -\cdot )}+\frac{e^{-\varepsilon\phi (x -\cdot )}-1}{\varepsilon}\right)=B\left(\theta -\phi(x-\cdot)\right),\quad a.a.~x\in\R^d.
\]
This shows the pointwise convergence of the integrand functions which 
appear in the definition of $(\tilde L_{\varepsilon, \mathrm{ren}}B)(\theta)$ and 
$(\tilde L_V B)(\theta)$. In addition, for all $\varepsilon >0$ we have
\[
\left|B\left(\theta e^{-\varepsilon\phi (x -\cdot )}+\frac{e^{-\varepsilon\phi (x -\cdot )}-1}{\varepsilon}\right)\right|\leq
\|B\|_\alpha\exp\left(\frac{1}{\alpha}\|\theta\|_{L^1}+\frac{1}{\alpha}\|\phi\|_{L^1}\right),
\]
leading through an application of the Lebesgue dominated convergence theorem 
to the required limit. 

(ii) In Lemma \ref{Lemma2} replace $\varphi$ by
$e^{-\varepsilon\phi}$ and $\psi$ by
$\frac{e^{-\varepsilon\phi}-1}{\varepsilon}$. Arguments similar to
prove Proposition \ref{Proposition2} together with Lemma
\ref{Lemma1} complete the proof for $\tilde{L}_{\varepsilon, \mathrm{ren}}$. 
For $\tilde{L}_V$, the proof follows similarly.
\end{proof}

Proposition \ref{estimativas} (ii) provides similar estimate of norms for 
$\tilde{L}_{\varepsilon, \mathrm{ren}}$, $\eps>0$, and the limiting mapping 
$\tilde{L}_V$, namely, $\|\tilde{L}_{\varepsilon, \mathrm{ren}}\|_{\alpha'},\|\tilde{L}_V\|_{\alpha'}\leq M\|B\|_{\alpha''}$, $0<\alpha\leq\alpha'<\alpha''\leq\alpha_0$, with
\[
M:=\frac{\alpha_{0}}{\alpha''-\alpha'}
\left(1+z\alpha_0
e^{\frac{\|\phi\|_{L^1}}{\alpha}-1}\right). 
\]
Therefore, given any 
$B_{0,V},B_{0,\mathrm{ren}}^{(\eps)}\in\mathcal{E}_{\alpha_0}$, $\eps>0$, it 
follows from Theorem \ref{Th13} and its proof that for each 
$s\in\left(0,s_0\right)$ and $\delta=\frac{1}{eM}$ there is a unique solution 
$B_{t,\mathrm{ren}}^{(\eps)}:\left[0,\delta(s_0-s)\right)\to\mathcal{E}_\alpha$, 
$\eps>0$, to each initial value problem \eqref{V12} and a unique solution 
$B_{t,V}:\left[0,\delta(s_0-s)\right)\to\mathcal{E}_\alpha$ to the initial 
value problem
\begin{equation}
\frac{\partial}{\partial t}B_{t,V}=\tilde L_{V}B_{t,V},\quad
{B_{t,V}}_{|t=0}=B_{0,V}.\label{V19}
\end{equation}
That is, independent of the initial value problem under consideration, the 
solutions obtained are defined on the same time-interval and with values in 
the same Banach space. Therefore, it is natural to analyze under which 
conditions the solutions to \eqref{V12} converges to the solution to
\eqref{V19}. This follows straightforwardly from a general result which proof 
(see Appendix) follows closely the lines of the proof of Theorem \ref{Th1}.

\begin{theorem}\label{Thconv} On a scale of Banach spaces 
$\{\B_s: 0<s\leq s_0\}$ consider a family of initial value problems
\begin{equation}
\frac{du_\eps(t)}{dt}=A_{\eps}u_\eps(t),\quad u_\eps(0)=u_{\eps}\in
\mathbb{B}_{s_0},\quad \eps\geq0,\label{V1eps}
\end{equation}
where, for each $s\in(0,s_0)$ fixed and for each pair $s', s''$ such that 
$s\leq s'<s''\leq s_0$, $A_\eps:\B_{s''}\to\B_{s'}$ is a linear mapping so that 
there is an $M>0$ such that for all $u\in\B_{s''}$
\[
\|A_{\eps}u\|_{s'}\leq\frac{M}{s''-s'}\|u\|_{s''}.
\]
Here $M$ is independent of $\eps, s',s''$ and $u$, however it might depend
continuously on $s,s_0$. Assume that there is a $p\in\N$ and for each 
$\eps>0$ there is an $N_\eps>0$ such that for each pair $s', s''$, 
$s\leq s'<s''\leq s_0$, and all $u\in\B_{s''}$
\[
\|A_{\eps}u-A_0u\|_{s'}\leq \sum_{k=1}^p\frac{N_\eps}{(s''-s')^k}\|u\|_{s''}.
\]
In addition, assume that $\lim_{\eps\rightarrow0}N_\eps=0$ and 
$\lim_{\eps\rightarrow0}\|u_{\eps}(0)-u_{0}(0)\|_{s_0}=0$.

Then, for each $s\in(0,s_0)$, there is a constant $\delta>0$ (which depends on 
$M$) such that there is a unique solution 
$u_\eps:\left[0,\delta(s_0-s)\right)\to\B_s$, $\eps\geq0$, to each initial 
value problem (\ref{V1eps}) and for all $t\in\left[0,\delta(s_0-s)\right)$ we 
have
\[
\lim_{\eps\rightarrow0}\|u_{\eps}(t)-u_{0}(t)\|_{s}=0.
\]
\end{theorem}

To proceed to an application of this general result one needs the following 
estimate of norms.

\begin{proposition}\label{Prop3}
Assume that $0\leq\phi\in L^1\cap L^\infty$ and let $\alpha_0>\alpha>0$ be 
given. Then, for all $B\in\mathcal{E}_{\alpha''}$, 
$\alpha''\in (\alpha, \alpha_0]$, the following estimate holds
 \begin{equation*}
\|\tilde L_{\varepsilon, \mathrm{ren}}B-\tilde L_V B\|_{\alpha'}\leq
\eps z\|\phi\|_{L^\infty}\|B\|_{\alpha''}e^{\frac{\|\phi\|_{L^1}}{\alpha}}\biggl(
  \frac{\|\phi\|_{L^1} \alpha_0}{\alpha''-\alpha'}+\frac{4\alpha_0^3}{(\alpha''-\alpha')^2e}\biggr)
\end{equation*}
for all $\alpha'$ such that $\alpha\leq\alpha'<\alpha''$ and all $\eps>0$.
\end{proposition} 

\begin{proof} First we observe that
\begin{eqnarray}
&&\left|(\tilde L_{\varepsilon, \mathrm{ren}}B)(\theta)-(\tilde L_V B)(\theta)\right|\leq z\int_{\R^d}dx\,\left|\theta(x)\right|\nonumber\\
&&\qquad\quad\times\left|  B\left(\theta e^{-\varepsilon\phi (x -\cdot )}+\frac{e^{-\varepsilon\phi (x -\cdot
)}-1}{\varepsilon}\right)-B\left(\theta -\phi (x -\cdot)\right)\right|.\label{Ola}
\end{eqnarray}
In order to estimate (\ref{Ola}), given any $\theta_1,\theta_2\in L^1$, let us 
consider the function 
$C_{\theta_1,\theta_2}(t)=B\left(t\theta_1+(1-t)\theta_2\right)$, 
$t\in\left[0,1\right]$. One has
\begin{align*}
\frac{\partial}{\partial
t} C_{\theta_1,\theta_2}(t)&=\frac{\partial}{\partial
s} C_{\theta_1,\theta_2}(t+s)\Bigr|_{s=0}=\frac{\partial}{\partial
s} B\bigl(\theta_2+t(\theta_1-\theta_2)+s(\theta_1-\theta_2)\bigr)\Bigr|_{s=0}\\&=\int_{\mathbb{R}^d}dx\,(\theta_1(x)-\theta_2(x))\,\vd B(\theta_2+t(\theta_1-\theta_2);x),
\end{align*}
leading to
\begin{align*}
\bigl|B(\theta_1)-B(\theta_2)\bigr|&=\bigl|C_{\theta_1,\theta_2}(1)-C_{\theta_1,\theta_2}(0)\bigr|
\\&\leq \max_{t\in[0,1]}\int_{\mathbb{R}^d}dx\,\bigl|\theta_1(x)-\theta_2
(x)\bigr| \bigl|\vd B(\theta_2+t(\theta_1-\theta_2);x)\bigr|\\&\leq\|\theta_1-\theta_2\|_{L^1}\max_{t\in[0,t]}\|\vd B(\theta_2+t(\theta_1-\theta_2);\cdot)\|_{L^\infty},
\end{align*}
where, through similar arguments to prove Lemma \ref{Lemma1},
\[
\bigl\|\vd B(\theta_2+t(\theta_1-\theta_2);\cdot)\bigr\|_{L^\infty}\leq\frac{e}{\alpha''}\|B\|_{\alpha''}\exp\left(\frac{\|\theta_2+t(\theta_1-\theta_2)\|_{L^1}}{\alpha''}\right).
\]
As a result
\[
\bigl|B(\theta_1)-B(\theta_2)\bigr|\leq \frac{e}{\alpha''}\|\theta_1-\theta_2\|_{L^1}\|B\|_{\alpha''}\max_{t\in[0,1]}\exp\left(\frac{t\|\theta_1\|_{L^1}+(1-t)\|\theta_2\|_{L^1}}{\alpha''}\right),
\]
for all $\theta_1,\theta_2\in L^1$. In particular, this shows that
\begin{align*}
&\left|  B\Bigl(\theta e^{-\varepsilon\phi (x -\cdot
)}+\frac{e^{-\varepsilon\phi (x -\cdot
)}-1}{\varepsilon}\Bigr)-B\bigl(\theta -\phi (x -\cdot
)\bigr)\right|\\\leq&
\eps\frac{e}{\alpha''}\|\phi\|_{L^\infty}\|B\|_{\alpha''}\left(\|\theta\|_{L^1}+
\|\phi\|_{L^1}\right)\\&\qquad\times
\max_{t\in[0,1]}\exp\left(\frac{1}{\alpha''}\left(t\left(\|\theta\|_{L^1}+
\|\phi\|_{L^1}\right)+(1-t)\left(\|\theta\|_{L^1}+\|\phi\|_{L^1}\right)\right)\right)\\=&\,\eps
\frac{e}{\alpha''}\|\phi\|_{L^\infty}\|B\|_{\alpha''}
\left(\|\theta\|_{L^1}+\|\phi\|_{L^1}\right)
\exp\left(\frac{1}{\alpha''}\left(\|\theta\|_{L^1}+\|\phi\|_{L^1}\right)\right),
\end{align*}
where we have used the inequalities
\begin{align*}
\|\theta e^{-\varepsilon\phi (x -\cdot )}-\theta\|_{L^1}&\leq\eps\|\phi\|_{L^\infty}\|\theta\|_{L^1},\\
\Bigl\|\frac{e^{-\varepsilon\phi (x -\cdot
)}-1}{\varepsilon}+\phi (x -\cdot
)\Bigr\|_{L^1}&\leq\eps\|\phi\|_{L^\infty}\|\phi\|_{L^1},\\
\Bigl\|\theta e^{-\varepsilon\phi (x -\cdot )}+\frac{e^{-\varepsilon\phi (x -\cdot
)}-1}{\varepsilon}\Bigr\|_{L^1}&\leq\|\theta\|_{L^1}+\|\phi\|_{L^1}.
\end{align*}

In this way we obtain
\begin{eqnarray*}
&&\|\tilde L_{\varepsilon, \mathrm{ren}}B-\tilde L_V B\|_{\alpha'}\\
&\leq&\,\eps \frac{ze}{\alpha''}\|\phi\|_{L^\infty}\|B\|_{\alpha''}
e^{\frac{\|\phi\|_{L^1}}{\alpha''}}\left\{\sup_{\theta\in L^1}\left(\|\theta\|_{L^1}^2
\exp\left(\|\theta\|_{L^1}\left(\frac{1}{\alpha''}-\frac{1}{\alpha'}\right)\right)\right)\right.\\
&&\left.+\|\phi\|_{L^1}\sup_{\theta\in L^1}\left(\|\theta\|_{L^1}\exp\left(\|\theta\|_{L^1}\left(\frac{1}{\alpha''}-\frac{1}{\alpha'}\right)\right)\right)\right\},
\end{eqnarray*}
and the proof follows using the inequalities 
$xe^{-mx}\leq\frac{1}{me}$ and $x^2e^{-mx}\leq\frac{4}{m^2e^2}$ for $x\geq0$, 
$m>0$. 
\end{proof}

We are now in conditions to state the following result.

\begin{theorem}\label{Props1} Given an $0<\alpha<\alpha_0$, let 
$B_{t,\mathrm{ren}}^{(\varepsilon)}, B_{t,V}$, 
$t\in\left[0,\delta(\alpha_0-\alpha)\right)$, be the 
local solutions in $\mathcal{E}_{\alpha}$ to the initial value problems 
\eqref{V12}, \eqref{V19} with $B_{0,\mathrm{ren}}^{(\varepsilon)},B_{0,V}\in\mathcal{E}_{\alpha_0}$. If $0\leq\phi\in L^1\cap L^\infty$ and
$\lim_{\varepsilon\rightarrow 0}\|B_{0,\mathrm{ren}}^{(\varepsilon)}-B_{0,V}\|_{\alpha_0}=0$, then, for each $t\in\left[0,\delta(\alpha_0-\alpha)\right)$, 
\[
\lim_{\varepsilon\rightarrow 0}\|B_{t,\mathrm{ren}}^{(\varepsilon)}-B_{t,V}\|_{\alpha}=0.
\]
Moreover, if 
$B_{0,V}(\theta)=\exp\left(\int_{\mathbb{R}^d}dx\,\rho_0(x)\theta(x)\right)$, 
$\theta\in L^1$, for some function $0\leq\rho_0\in L^\infty$ such that
$\|\rho_0\|_{L^\infty}\leq \frac{1}{\alpha_0}$, and 
$\max\{\frac{1}{\alpha_0},z\}<\frac{1}{\alpha}$ then, for each
$t\in\left[0,\delta(\alpha_0-\alpha)\right)$,
\begin{equation}\label{expsol}
B_{t,V}(\theta)=\exp\left(\int_{\mathbb{R}^d}dx\,
\rho_t(x)\theta(x)\right), \quad \theta\in L^1,
\end{equation}
where $0\leq\rho_t\in L^\infty$ is a classical solution to the equation
\eqref{Vl42} such that, for each $t\in\left[0,\delta(\alpha_0-\alpha)\right)$, 
$\|\rho_t\|_{L^\infty}\leq\frac{1}{\alpha}$.
\end{theorem}

\begin{proof}
The first part follows directly from Proposition~\ref{Prop3} and 
Theorem~\ref{Thconv} for $p=2$ and $N_\eps=\eps z\|\phi\|_{L^\infty}\alpha_0
e^{\frac{\|\phi\|_{L^1}}{\alpha}}\max\bigl\{\|\phi\|_{L^1},\frac{4\alpha_0^2}{e}\bigr\}$.

Concerning the last part, we begin by observing that it has been shown in 
\cite[Proof of Theorem 3.3]{FKK10} that given a $\rho_0\in L^\infty$ such that
$\|\rho_0\|_{L^\infty}\leq\frac{1}{\alpha_0}$, the solution $\rho_t$ to 
\eqref{Vl42} (which existence has been proved in \cite{FKKoz2011})
fulfills $\|\rho_t\|_{L^\infty}\leq\max\{\frac{1}{\alpha_0},z\}$. In this way, 
the assumption $\max\{\frac{1}{\alpha_0},z\}<\frac{1}{\alpha}$ implies that
$B_{t,V}$, given by \eqref{expsol}, fulfills $B_{t,V}\in\mathcal{E}_\alpha$. 
Then, by an argument of uniqueness, to prove the last assertion amounts to 
show that $B_{t,V}$ solves equation \eqref{V19}. For this purpose we note that 
for any $\theta,\theta_1\in L^1$ we have
\[
\frac{\partial}{\partial z_1}
B_{t,V}(\theta+z_1\theta_1)\biggr|_{z_1=0}=B_{t,V}(\theta)\int_{\mathbb{R}^d}dx
\rho_t(x)\theta_1(x),
\]
and thus $\vd B_{t,V}(\theta;x)= B_{t,V}(\theta) \rho_t(x)$. Hence, for all 
$\theta\in L^1$,
\[
(\tilde L_V B_{t,V})(\theta)=
-B_{t,V}(\theta)\left(\int_{\R^d}dx\,\theta(x)\rho_t(x) - z
\int_{\R^d}dx\,\theta(x) \exp\left(-(\rho_t*\phi)(x)\right)\right).
\]
Since $\rho_t$ is a classical solution to \eqref{Vl42}, $\rho_t$ solves a weak 
form of equation \eqref{Vl42}, that is, the right-hand side of the latter 
equality is equal to
\[
B_{t,V}(\theta)\frac{d}{dt}\int_{\mathbb{R}^d}dx\,
\rho_t(x)\theta(x) = \frac{\partial}{\partial t}
B_{t,V}(\theta).\qedhere
\]
\end{proof}

\section*{Appendix: Proofs of Theorems \ref{Th1} and
\ref{Thconv}}

\begin{proof}[Sketch of the proof of Theorem \ref{Th1}.]
For some $t>0$ which later on will be properly chosen, let us consider the 
sequence of functions $(u_n)_{n\in\mathbb{N}_0}$ with 
$u_{0}(t)\equiv u_0\in\mathbb{B}_{s_0}$ and
\begin{equation*}
u_n(t):= u_0 + \int_0^t (Au_{n-1})(s) d s, \quad n \in \mathbb{N}.
\end{equation*}
By an induction argument, it is easy to check that $u_n(t)\in\mathbb{B}_s$ 
for any $s<s_0$ and, in an equivalent way, the sequence may be rewritten as   
\begin{equation}
\label{8}
u_n(t) = u_0 + \sum_{m=1}^n \frac{t^m}{m!}A^mu_0.
\end{equation}

Fixed an $0<s<s_0$, let us now consider a partition of the interval 
$\left[s,s_0\right]$ into $m$ equals parts, $m\in\mathbb{N}$. That 
is, we define $s_l := s_{0}-\frac{l(s_{0} -s)}{m}$ for $l = 0, \dots , m$. By 
assumption, observe that for each $l = 0, \dots , m$ the linear mapping 
$A:\mathbb{B}_{s_{l}}\to\mathbb{B}_{s_{l+1}}$ verifies
\[
\|A\|_{s_{l}s_{l+1}}:=\|A\|_{\mathbb{B}_{s_{l}}\mapsto\mathbb{B}_{s_{l+1}}}\leq\frac{mM}{s_0-s},
\]
and thus
\begin{equation}
 \label{10C}
 \| A^m\|_{ s_0 s}  \leq \|A\|_{s_0s_1} \cdots \|A\|_{s_{m-1}s}
 \leq \left(\frac{mM}{s_0-s}\right)^m.
\end{equation}
From this and the Stirling formula follow the convergence of the series
\[
\sum_{m=1}^n \frac{t^m}{m!}\|A^mu_0\|_s\leq 
\|u_0\|_{s_0}\sum_{m=1}^n \frac{m^m}{m!}\left(\frac{Mt}{s_0-s}\right)^m
\]
whenever $\frac{tM}{s_0-s}<\frac{1}{e}$. This means that for all
$t<\frac{s_0-s}{eM}$ the sequence \eqref{8} converges in $\B_s$ to the 
function 
\begin{equation*}
u(t) := u_0 + \sum_{m=1}^\infty \frac{1}{m!}t^m A^mu_0.
\end{equation*} 
Moreover, setting $\delta:=\frac{1}{e M}$, $M=M(s,s_0)$, this convergence is 
uniform on any interval $[0,T]\subset[0,\delta(s_0-s))$. Similar arguments 
show that an analogous situation occurs for the series 
\begin{equation}\label{der}
\sum_{m=1}^\infty \frac{1}{m!}\frac{d}{dt}t^m A^mu_0=\sum_{m=0}^\infty \frac{1}{m!}t^mA^{m+1}u_0.
\end{equation}
This shows that on the time-interval $(0,\delta(s_0-s))$ $u$ is a continuously 
differentiable function in $\B_s$. 

Of course, these considerations hold for any $s_1\in\left(s,s_0\right)$, 
showing that the sequence \eqref{8} also converges in the space $\B_{s_1}$ 
uniformly to a function $\tilde u$ on any time interval 
$[0,T]\subset[0,\delta_1(s_0-s_1))$, $\delta_1:=\frac{1}{eM_1}$, 
$M_1=M_1(s_0,s_1)$. On the other hand, due to the continuity of $M(s_0,\cdot)$ 
on $(0,s_0)$, for each $t\in[0,\delta(s_0-s))$ fixed there is an 
$s_1\in(s,s_0)$ such that $t\in[0,\delta_1(s_0-s_1))$. As a result, $u_n(t)$ 
converges to a $\tilde{u}(t)$ in the space $\B_{s_1}\subset\B_s$. Since
\[
\|\tilde{u}(t)-u(t)\|_{\B_s}\leq\|\tilde{u}(t)-u_n(t)\|_{\B_{s_1}}+\|u(t)-u_n(t)\|_{\B_s},
\]
it follows that $\tilde{u}(t)=u(t)$ in $\B_s$. In other words, 
$u(t)\in\B_{s_1}$. Therefore, $u(t)$ is in the domain of $A:\B_{s_1}\to\B_s$, 
and thus $Au(t)\in\B_s$. Since this holds for every $t\in[0,\delta(s_0-s))$, 
the convergence of the series \eqref{der} then implies that $u$ is a solution 
to the initial value problem (\ref{V1}). To check the uniqueness see 
e.g.~\cite[pp.~16--17]{T68}.
\end{proof}

\begin{proof}[Proof of Theorem \ref{Thconv}.] To prove this result amounts to 
check the convergence. Following the scheme used to prove Theorem \ref{Th1}, 
we begin by recalling that in that proof each solution $u_\eps$, $\eps\geq 0$, 
to \eqref{V1eps} was obtained as a limit in $\B_s$ of
\begin{equation*}
u_{\eps,n}(t) = u_\eps + \sum_{m=1}^n \frac{1}{m!}t^mA_\eps^m u_\eps,
\end{equation*}
where $t\in[0,\delta(s_0-s))$ with $\delta=\frac{1}{eM}$. Thus, for each
$\eps'>0$, there is an $n\in\mathbb{N}$ such that
\begin{align}\nonumber
\|u_\eps(t)
-u_0(t)\|_s\leq&\,\|u_\eps(t)-u_{\eps,n}(t)\|_s+\|u_{\eps,n}(t)-u_{0,n}(t)\|_s+\|u_{0,n}(t
)-u_0(t)\|_s\\<&\,\frac{\eps'}{2}+\|u_\eps-u_0\|_s+\sum_{m=1}^n\frac{t^m}{m!}
\|A_\eps^mu_\eps-A_0^mu_0\|_s\nonumber\\\leq&\,\frac{\eps'}{2}+\|u_\eps-u_0\|_s+\sum_{m=1}^n\frac{t^m}{m!}
\|A_\eps^m(u_\eps-u_0)\|_s\nonumber\\&+\sum_{m=1}^n\frac{t^m}{m!}
\|(A_\eps^m-A_{0}^m)u_0)\|_s.\label{est}
\end{align}
Observe that by \eqref{10C}
\[
\|A_\eps^m(u_\eps-u_0)\|_s\leq\biggl(\frac{mM}{s_0-s}\biggr)^m\|u_\eps-u_0\|_{s_0}.
\]
To estimate \eqref{est} we proceed as in the proof of Theorem \ref{Th1}. For 
this purpose, we will use the decomposition
\begin{align}
 \label{U2}
  A_\eps^m - A_0^m = & \left( A_\eps- A_0 \right)  A_\eps^{m-1} + A_0
\left( A_\eps - A_0 \right)  A_\eps^{m-2}+ \nonumber\\
 &  + \cdots + A_0^{m-2} \left( A_\eps- A_0 \right)
  A_\eps + A_0^{m-1} \left( A_\eps- A_0 \right) .\nonumber
\end{align}
Then, considering again a partition of the interval $\left[s,s_0\right]$ 
into $m$ parts and the points $s_l = s_{0}-\frac{l(s_{0} -s)}{m}$, 
$l = 0, \dots , m$, one finds the estimate 
\begin{align*}
\|(A_\eps^m-A_{0}^m)u_0)\|_s&\leq\sum_{l=0}^{m-1}
\|A_\eps-A_0\|_{s_ls_{l+1}}\biggl(\frac{mM}{s_0-s}\biggr)^{m-1}\|u_0\|_{s_0}\\
&\leq\sum_{k=1}^p\frac{N_\eps}{(s_0-s)^{k-1}}\biggl(\frac{mM}{s_0-s}\biggr)^m\frac{m^k}{M}\|u_0\|_{s_0}.
\end{align*}
As a result, defining for each $t\in\bigl[0,\delta (s_0-s)\bigr)$, 
$\delta=\frac{1}{eM}$,
\[
f_q(t):=\sum_{m=1}^\infty
\frac{m^q}{m!}\biggl(\frac{tmM}{s_0-s}\biggr)^m<\infty, \quad q\geq0,
\] 
we obtain from the previous considerations the estimate
\[
\|u_\eps(t)-u_0(t)\|_s<\frac{\eps'}{2}+\|u_\eps-u_0\|_s+\|u_\eps-u_0\|_{s_0}f_0(t)+\sum_{k=1}^p\frac{N_\eps}{M(s_0-s)^{k-1}}\|u_0\|_{s_0}f_k(t).
\]
Here we observe that, by assumption, $u_\eps$ converges in $\B_{s_0}$ to 
$u_0$. Thus, by the definition of a scale of Banach spaces, this convergence 
also holds in $\B_s$. Therefore, for small enough $\eps$, one has 
$\|u_\eps(t) -u_0(t)\|_s<\eps'$, which completes the proof.
\end{proof}

\subsection*{Acknowledgments}

Financial support of DFG through SFB 701 (Bielefeld University) and FCT through
PTDC/MAT/100983/2008 and ISFL-1-209 are gratefully acknowledged.


\end{document}